\documentclass{article}
\usepackage[utf8x]{inputenc}
\usepackage{amsmath,amssymb,amsthm}
\usepackage{graphicx, color, xcolor}
\usepackage{algorithm}
\usepackage{algpseudocode}
\usepackage{relsize}
\usepackage[mathscr]{eucal}
\usepackage{paralist}
\usepackage{url}
\usepackage{hyperref}
\usepackage{times}
\usepackage{tikz}
\usepackage{xspace}
\usetikzlibrary{backgrounds, positioning, fit, patterns}
\usepackage{booktabs}
\usepackage{threeparttable}
\usepackage[margin  =  3 cm]{geometry}

\newcommand{\MyComment}[1]{\Comment{\textcolor{blue}{#1}}}

\newcommand{\defeq}{\stackrel{\mathsmaller{\mathsf{def}}}{=}}
\newcommand{\Prob}[1]{\mathrm{Pr}\left[#1\right]\xspace}
\newcommand{\One}{\text{One}\xspace}
\newcommand{\disjoint}{C_1\! \nleftrightarrow\! C_2\xspace}

\newtheorem{theorem}{Theorem}
\newtheorem{definition}{Definition}
\newtheorem{lemma}{Lemma}

\theoremstyle{remark}
\newtheorem{remark}{Remark}
\newtheorem{obs}{Observation}

\def\elected{\mbox{\small ELECTED}}
\def\nonelected{\mbox{\small NON-ELECTED}}

\title{The Complexity of Leader Election: A Chasm at Diameter Two}

\author
{
    Soumyottam Chatterjee\thanks{Department of Computer Science, University of Houston, Houston, TX 77204, USA. Email: \texttt{schatterjee4@uh.edu, gopal@cs.uh.edu}. Supported, in part, by NSF grants CCF-1527867, CCF-1540512,  IIS-1633720,  and CCF-BSF-1717075.} \and Gopal Pandurangan$^*$ \and Peter Robinson \thanks{Department of Computing \& Software, McMaster University, Hamilton, Ontario L8S 4K1, Canada. Email: \texttt{peter.robinson@mcmaster.ca}. Supported, in part, by the Natural Sciences and Engineering Research Council of Canada (NSERC), RGPIN-2018-06322.}
}

\begin{document}

\maketitle
\thispagestyle{empty}
\begin{abstract}
This paper focuses on studying the message complexity of implicit leader election in synchronous  distributed networks of diameter two. Kutten et al.\ [JACM 2015] showed a fundamental lower bound of $\Omega(m)$ ($m$ is the number of edges in the network) on the message complexity of (implicit) leader election that applied also to Monte Carlo randomized algorithms with constant success probability; this lower bound applies for graphs that have diameter at least three. On the other hand, for complete graphs (i.e., graphs with diameter one), Kutten et al.\ [TCS 2015] established a tight bound of $\tilde{\Theta}(\sqrt{n})$ on the message complexity of randomized leader election ($n$ is the number of nodes in the network). For graphs of diameter two, the complexity was not known. 

In this paper, we settle this complexity by showing a tight bound of $\tilde{\Theta}(n)$ on the message complexity of leader election in diameter-two networks. We first give a simple randomized Monte-Carlo leader election algorithm that with high probability (i.e., probability at least $1 - n^{-c}$, for some fixed positive constant $c$) succeeds and  uses $O(n\log^3{n})$ messages and runs in $O(1)$ rounds; this algorithm works without knowledge of $n$ (and hence needs no global knowledge). We then show that any algorithm (even Monte Carlo randomized algorithms with large enough constant success probability) needs $\Omega(n)$ messages (even when $n$ is known), regardless of the number of rounds. We also present an $O(n\log{n})$ message deterministic algorithm that takes $O(\log{n})$ rounds (but needs knowledge of $n$); we show that this message complexity is tight for deterministic algorithms.

Our results show that leader election can be solved in diameter-two graphs with (essentially) linear (in $n$) message complexity and thus the $\Omega(m)$ lower bound does not apply to diameter-two graphs. Together with the two previous results of Kutten et al., our results fully characterize the message complexity of leader election vis-\`a-vis the graph diameter.\\

\end{abstract}

\newpage

\setcounter{page}{1}


\section{Introduction} \label{sec:intro}

Leader election is a classical and  fundamental problem in distributed computing. The leader election problem requires a group of processors in a distributed network to elect a unique leader among themselves, i.e., exactly one processor must output the decision that it is the leader, say, by changing a special \emph{status} component of its state to the value \emph{leader}~\cite{Lynch_1996_Book}. All the rest of the nodes must change their status component to the value \emph{non-leader}. These nodes need not be aware of the identity of the leader. This {\em implicit} variant of leader election is quite standard (cf. \cite{Lynch_1996_Book}), and has been extensively studied (see e.g., \cite{Kutten_2015_JACM} and the references therein) and is sufficient in many applications, e.g., for token generation in a token ring environment~\cite{Lann_1977}. In this paper, we focus on this implicit variant.
\footnote{In another variant, called {\em explicit} leader election,  all the non-leaders change their status component to the value \emph{non-leader}, and moreover, every node must also know the identity of the unique leader. In this variant,  $\Omega(n)$ messages is an obvious lower bound (throughout, $n$ denotes the number of nodes in the network) since every node must be informed of the leader's identity. Clearly, any lower bound for implicit leader election applies to explicit leader election as well.}

The complexity of leader election, in particular, its message and time complexity, has been extensively studied both in general graphs as well as in special graph classes such as rings and complete networks, see e.g., \cite{Lynch_1996_Book, Peleg_1990, Santoro_2006_Book, Tel_2001_Book, Kutten_2015_TCS, Kutten_2015_JACM}. While much of the earlier work focused on deterministic algorithms, recent works have studied randomized algorithms (see e.g., \cite{Kutten_2015_TCS, Kutten_2015_JACM} and the references therein). Kutten et al.\ \cite{Kutten_2015_JACM} showed a fundamental lower bound of $\Omega(m)$  ($m$ is the number of edges in the network) on the message complexity of (implicit) leader election that applied even to Monte Carlo randomized algorithms with (large-enough) constant success probability; this lower bound applies for graphs {\em that have diameter at least three}. We point that the $\Omega(m)$ lower bound applies even for algorithms that have knowledge of $n$, $m$, $D$ (throughout, $n$ denotes the number of nodes, $m$ the number of edges, and $D$ the network diameter). The lower bound proof involves constructing a ``dumb-bell" graph $G$ which consists of two regular subgraphs $G_1$ and $G_2$ (each having approximately $\frac{m}{2}$ edges) joined by a couple of ``bridge" edges (the bridge edges are added so that the regularity is preserved). Note that (even) if $G_1$ and $G_2$ are cliques (in particular, they can be any 2-connected graph) then $G$ will be of diameter (at least) three. This is the smallest diameter that makes the lower bound proof work; we refer to \cite{Kutten_2015_JACM} for details.

On the other hand, for complete graphs (i.e., diameter one), Kutten et al. \cite{Kutten_2015_TCS} established a tight bound of $\tilde{\Theta}(\sqrt{n})$ on the message complexity of randomized leader election ($n$ is the number of nodes in the network). In other words, they showed an $\tilde{O}(\sqrt{n})$ messages algorithm that elects a (unique) leader with high probability. To complement this, they also showed that any leader election algorithm in a complete graph requires $\tilde{\Omega}(\sqrt{n})$ messages to succeed with (large-enough) constant probability.

For graphs of diameter two, the message complexity was not known. In this paper, we settle this complexity by showing a tight bound of $\tilde{\Theta}(n)$ on the message complexity of leader election in diameter-two networks. In particular,  we present a simple randomized leader election algorithm that takes $O(n\log^3{n})$ messages and $O(1)$ rounds that works {\em even when $n$ is not known}. In contrast, we show that any randomized algorithm (even Monte Carlo algorithms with constant success probability) needs $\Omega(n)$ messages. Our results show that leader election can be solved in diameter-two graphs in (essentially) linear (in $n$) message complexity which is optimal (up to a $\text{polylog}(n)$ factor) and thus the $\Omega(m)$ message lower bound does not apply to diameter-two graphs. Together with the previous results \cite{Kutten_2015_JACM, Kutten_2015_TCS}, our results fully characterize the message complexity of leader election vis-\`a-vis the graph diameter (see Table \ref{table1}).
\subsection{Our Results} \label{sec:results}

This paper focuses on studying the message complexity of leader election (both randomized and deterministic) in synchronous distributed networks, in particular, in  networks of {\em diameter two}.

For our algorithms, we assume that the communication is {\em synchronous} and follows the standard $\mathcal{CONGEST}$ model~\cite{Peleg_2000_Book}, where a node can send in each round at most one message of size $O(\log{n})$ bits on a single edge.  We assume that the nodes have unique IDs.  We assume that all nodes wake up simultaneously at the beginning of the execution. (Additional details on our distributed computation model are given in Section \ref{sec:model}.)

We show the following results:
\begin{enumerate}

    \item \textbf{Algorithms:} We show that the message complexity of leader election in diameter-two graphs is $\tilde{O}(n)$, by presenting a randomized (implicit) leader election algorithm    (cf.\ Section \ref{sec:random}), that takes $O(n\log^3{n})$ messages and runs in $O(1)$ rounds with high probability (whp).
    \footnote{Throughout, ``with high probability" means with probability at least $1 - n^{-c}$, for some fixed positive constant $c$.}
    This algorithm works even without knowledge of $n$. While it is easy to design an $O(n\log{n})$ messages randomized algorithm with knowledge of $n$ (Section \ref{sec:technical-overview}), not having knowledge of $n$ makes the analysis more involved.

    We also present a {\em deterministic}  algorithm that uses only $O(n\log{n})$ messages, but takes $O(\log{n})$ rounds. Also this algorithm needs knowledge of $n$ (or at least a constant factor upper bound of $\log{n}$) (cf.\ Section \ref{sec:deter}).

    We note that all our algorithms will work seamlessly for complete networks as well.

    \item \textbf{Lower Bounds:} We show that, in general, it is not possible to improve over our algorithm substantially, by presenting a lower bound for leader election that applies also to randomized (Monte Carlo) algorithms. We show that $\Omega(n)$ messages are needed for any leader election algorithm (regardless of the number of rounds) in a diameter-two network which succeeds with any constant probability that is strictly larger than $\frac{1}{2}$ (cf.\ Section \ref{sec:lb}). This lower bound holds even in the $\mathcal{LOCAL}$ model ~\cite{Peleg_2000_Book}, where there is no restriction on the number of bits that can be sent on each edge in each round. To the best of our knowledge, this is the first non-trivial lower bound for randomized leader election in diameter-two networks.

    We also show a simple deterministic reduction that shows that any super-linear message lower bound for complete networks also applies to diameter-two networks as well (cf. Section \ref{sec:detlb}). It can be shown that $\Omega(n \log{n})$ messages is a lower bound for deterministic leader election in complete networks \cite{Afek_1991, Kutten_2017_private_communication} (under the assumption that the number of rounds is bounded by some function of $n$).
    \footnote{Afek and Gafni\cite{Afek_1991} show the $\Omega(n\log{n})$ message lower bound for complete networks under the non-simultaneous wakeup model in synchronous networks. The same message bound can be shown to hold in the simultaneous wake-up model as well under the restriction that the number of rounds is bounded by a function of $n$ \cite{Kutten_2017_private_communication}.}
    By our reduction this lower bound also applies for diameter-two networks.
    \footnote{We point out that lower bounds for complete networks do not directly translate to diameter-two networks.}

\end{enumerate}


\begin{threeparttable}
  \begin{tabular}{l l l l l}
  \toprule
   & \multicolumn{2}{l}{\textsc{Randomized}} & \multicolumn{2}{c}{\textsc{Deterministic}} \\
  Diameter & Time & Messages & Time & Messages \\
  \midrule
  {\boldmath $D = 1$:} \cite{Kutten_2015_TCS, Afek_1991}\\
  Upper Bound & $O(1)$ & $O(\sqrt{n}\log^{\frac{3}{2}}\!n)$ & $O(1)$$^\dagger$  & $O(n\log n)$ $^\dagger$\\
  Lower Bound & $\Omega(1)$ & $\Omega(\sqrt{n})$ & $\Omega(1)$  & $\Omega(n\log n)$ \\
  \midrule
  {\boldmath $ D = 2$:} & \multicolumn{4}{c}{\bf Our Results}\\
  Upper Bound & $O(1)$ & $O(n \log^3{n})$ & $O(\log n)$$^{\dagger\dagger}$ & $O(n\log n)$\tnote{$\$$} \\
  Lower Bound & $\Omega(1)$ & $\Omega(n)$ & $\Omega(1)$ & $\Omega(n\log n)$ \\
  \midrule
  {\boldmath $D \ge 3$:} \cite{Kutten_2015_JACM}\\
  Upper Bound & $O(D)$ & $O(m \log\log n)$ & $O(D \log n)$ & $O(m \log n)$ \\
  Lower Bound & $\Omega(D)$ & $\Omega(m)$ & $\Omega(D)$ & $\Omega(m)$ \\
  \bottomrule
  \end{tabular}
  \begin{tablenotes}
  \item[$\dagger$] Note that attaining $O(1)$ time requires $\Omega(n^{1 + \Omega(1)})$ messages in cliques, whereas achieving $O(n\log{n})$ messages requires $\Omega(\log{n})$ rounds; see \cite{Afek_1991}.
  \item[\$] Needs knowledge of $n$.
  \item[$\dagger\dagger$] Note that it is easy to show a $O(1)$ round deterministic algorithm that takes $O(m)$ messages. 
  \end{tablenotes}
  \caption{Message and time complexity of leader election.} 
  \label{table1}
\end{threeparttable}

\subsection{Technical Overview}\label{sec:technical-overview}

All our algorithms exploit the following simple ``neighborhood intersection" property of diameter-two graphs: Any two nodes (that are non-neighbors) have at least one neighbor in common (please refer to Observation \ref{main-diameter-two-property}). However, note that unlike complete networks (which have been extensively studied with respect to leader election --- cf.\ Section \ref{sec:related}), in diameter-two networks, nodes generally do not have knowledge of $n$, the network size (in a complete graph, this is trivially known by the degree). This complicates obtaining sublinear in $m$ (where $m$ is the number of edges) message algorithms that are fully localized (don't have knowledge of $n$). Indeed, if $n$ is known, the following is a simple randomized algorithm:  each node becomes a candidate with probability $\Theta(\frac{\log{n}}{n})$ and sends its ID to all its neighbors; any node that gets one or more messages acts as a ``referee" and notifies the candidate that has the smallest ID (among those it has received). The neighborhood intersection property implies that at least one candidate  will be chosen uniquely as the leader with high probability.

If $n$ is not known, the above idea does not work. However, we show that if each node $v$ becomes a candidate with probability $\frac{1+ \log{d_v}}{d_v}$, (where $d_v$ is the degree of $v$) then the above idea can be made to work. The main  technical difficulty is then showing that at least one candidate is present (cf.\ Section \ref{sec:number-of-candidates}) and in bounding the message complexity (cf.\ Section \ref{sec:analysis-message-complexity}). We use Lagrangian optimization to prove that on expectation at least $\Theta(\log{n})$ candidates will be selected and then use a Chernoff bound to show a high probability result.

Our $\Omega(n)$ randomized lower bound is inspired by the {\em bridge crossing} argument of \cite{Kutten_2015_JACM} and \cite{Pai_2017}. In \cite{Kutten_2015_JACM}, the authors construct a ``dumbbell" graph $G$ which is done by taking two identical regular graphs $G_1$ and $G_2$, removing an edge from each and adding them as bridge edges between $G_1$ and $G_2$ (so that regularity is preservered). The argument is that any leader election algorithm should send at least one message across one of the two bridge edges (bridge crossing); otherwise, it can be shown that the executions in $G_1$ and $G_2$ are identical  leading to election of two leaders which is not valid. The argument in \cite{Kutten_2015_JACM} shows that $\Omega(m)$ messages are needed for bridge crossing. As pointed out earlier in Section \ref{sec:intro}, this construction makes the diameter of $G$ at least three and hence does not work for diameter-two graphs. To overcome this, we modify the construction that takes two complete graphs and add a set of bridge edges (as opposed to just two); see Fig~\ref{fig:messagelb}. This creates a diameter-two graph; however, the large number of bridge edges requires a different style of argument and results in a bound different compared to \cite{Kutten_2015_JACM}. We show that $\Omega(n)$ messages (in expectation) are needed to send a message across at least one bridge.

We also present a \emph{deterministic} algorithm that uses $O(n\log{n})$ messages, but takes $O(\log{n})$ rounds. Note that, in a sense, this improves over the randomized algorithm that sends $O(n\log^3{n})$ messages (although, we did not strive to optimize the $\log$ factors). However, the deterministic algorithm is slower by a $\log(n)$-factor and is more involved compared to the very simple randomized algorithm (although its analysis is a bit more complicated). Our deterministic algorithm uses ideas similar to Afek and Gafni's \cite{Afek_1991} leader election algorithm for complete graphs; however, the algorithm is a bit more involved. Our algorithm assumes knowledge of $n$ (this is trivially true in complete networks, since every node can infer $n$ from its degree) which is needed for termination. It is not clear if one can design an $O(n\log{n})$ messages algorithm (running in say $O(\log{n})$ rounds) that does not need knowledge of $n$, which is an interesting open question (cf.\ Section \ref{sec:conc}).

Finally, we present a simple reduction that shows that superlinear (in $n$) lower bounds in complete networks also imply lower bounds for diameter-two networks, by showing how using only $O(n)$ messages and in $O(1)$ rounds, a complete network can be converted to a diameter-two network in a distributed manner. This shows that our deterministic algorithm (cf.\ Section \ref{sec:deter}) is message optimal.
\subsection{Distributed Computing Model} \label{sec:model}

The model we consider is similar to the models of~\cite{Afek_1991, Humblet_1984, Korach_1990, Korach_1987, Korach_1989}, with the main addition of giving processors access to a private unbiased coin. We consider a system of $n$ nodes, represented as an undirected  graph $G = (V, E)$. In this paper, we focus on graphs with diameter $D(G) = 2$, where $D(G)$ is the diameter of $G = (V, E)$. An obvious consequence of this is that $G$ is connected, therefore $n - 1  \leq  m  \leq  \frac{n(n-1)}{2}$, where $m = |E|$ and $n = |V|$.

Each node has a unique identifier (ID) of $O(\log{n})$ bits and runs an instance of a distributed algorithm. The computation advances in synchronous rounds where, in every round, nodes can send messages, receive messages that were sent in the same round by neighbors in $G$, and perform some local computation. Every node has access to the outcome of unbiased private coin flips (for randomized algorithms). Messages are the only means of communication; in particular, nodes cannot access the coin flips of other nodes, and do not share any memory. Throughout this paper, we assume that all nodes are awake initially and simultaneously start executing the algorithm. We note that initially nodes have knowledge only of themselves, in other words we assume the {\em clean network model} --- also called the {\em KT0 model} \cite{Peleg_2000_Book} which is standard and most commonly used.
\footnote{If one assumes the {\em KT1 model}, where nodes have an initial knowledge of the IDs of their neighbors, there exists a trivial algorithm for leader election in a diameter-two graph that uses only $O(n)$ messages.}
\subsection{Leader Election: Problem Definition}

We formally define the leader election problem here.

Every node $u$ has a special variable $\texttt{status}_u$ that it can set to a value in
\begin{center}
	$\{\bot, \nonelected, \elected \}$;
\end{center}
initially we assume $\texttt{status}_u = \bot$.

An \emph{algorithm $A$ solves leader election in $T$ rounds} if, from round $T$ on, exactly one node has its status set to $\elected$ while all other nodes are in state $\nonelected$. This is the requirement for standard (implicit) leader election. For {\em explicit} leader election, we further require that all non-leader nodes should know the identity of the leader.
\subsection{Other Related Works} \label{sec:related}

The complexity of the leader election problem and algorithms for it, especially deterministic algorithms (guaranteed to always succeed), have been well-studied. Various algorithms and lower bounds are known in different models with synchronous (as well as asynchronous) communication and in networks of varying topologies such as a cycle, a complete graph, or some arbitrary topology (e.g., see \cite{Khan_2012, Lynch_1996_Book, Peleg_1990, Santoro_2006_Book, Tel_2001_Book, Kutten_2015_TCS, Kutten_2015_JACM} and the references therein).

The study of leader election algorithms is usually concerned with both message and time complexity. We discuss two sets of results, one for complete graphs and the other for general graphs. As mentioned earlier, for complete graphs, Kutten et al. \cite{Kutten_2015_TCS} showed that $\tilde{\Theta}(\sqrt{n})$ is the tight message complexity bound for randomized (implicit) leader election. In particular, they presented an $O(\sqrt{n}\log^{3/2}{n})$  messages algorithm that ran in $O(1)$ rounds; they also showed an almost matching lower bound for randomized leader election, showing that $\Omega(\sqrt n)$ messages are needed for  any leader election algorithm that succeeds with a sufficiently large constant probability.

For deterministic algorithms on complete graphs, it is known that $\Theta(n\log{n})$ is a tight bound on the message complexity ~\cite{Afek_1991, Kutten_2017_private_communication}. In particular, Afek and Gafni \cite{Afek_1991} presented an $O(n\log{n})$ messages algorithm for complete graphs that ran in $O(\log{n})$ rounds. For complete graphs, Korach et al.~\cite{Korach_1984} and Humblet \cite{Humblet_1984} also presented $O(n\log{n})$ message algorithms. Afek and Gafni \cite{Afek_1991} presented asynchronous and synchronous algorithms, as well as a tradeoff between the message and the time complexity of synchronous {\em deterministic} algorithms for complete graphs: the results varied from a $O(1)$-time, $O(n^2)$-messages  algorithm to a $O(\log{n})$-time, $O(n\log{n})$-messages algorithm. Afek and Gafni \cite{Afek_1991}, as well as~\cite{Korach_1984, Korach_1989} showed a lower bound of $\Omega(n\log{n})$ messages for {\em deterministic} algorithms in the general case.
\footnote{This lower bound assumes non-simultaneous wakeup though. If nodes  are assured to wake up at the same time in synchronous complete networks, there exists a trivial algorithm: if a node's identity is some $i$, it waits $i$ time before it sends any message then leader election could be solved (deterministically) in $O(n)$ messages on complete graphs in synchronous networks. Recently Kutten \cite{Kutten_2017_private_communication} shows that the $\Omega(n\log{n})$ lower bound holds for simulataneous wakeup as well, if the number of rounds is bounded.}

For general graphs, the best known bounds are as follows. Kutten et al.\ \cite{Kutten_2015_JACM} showed that $\Omega(m)$ is a very general lower bound on the number of messages and $\Omega(D)$ is a lower bound on the number of rounds for any leader election algorithm. It is important to point out that their lower bounds applied for graphs with {\em diameter at least three}. Note that these lower bounds hold even for randomized Monte Carlo algorithms that succeed even with (some large enough, but) constant success probability  and apply even for implicit leader election. Earlier results, showed such lower bounds only for deterministic algorithms and only for the restricted case of comparison algorithms, where it was also required that nodes may not wake up spontaneously and that $D$ and $n$ were not known. The $\Omega(m)$ and $\Omega(D)$ lower bounds are {\em universal} in the sense that they hold for all universal algorithms (namely, algorithms that work for all graphs), apply to  every $D \geq 3$, $m$, and $n$, and hold even if $D$, $m$, and $n$ are known, all the nodes wake up simultaneously, and the algorithms can make any use of node's identities. To show that these bounds are tight, they also  present an $O(m)$ messages algorithm (this algorithm is not time-optimal). An $O(D)$ time leader election algorithm is known \cite{Peleg_1990} (this algorithm is not message-optimal). They also presented an $O(m\log{\log{n}})$ messages randomized algorithm that ran in $O(D)$ rounds (where $D$ is the network diameter) that is simultaneously almost optimal with respect to both messages and time. They also presented an $O(m\log{n})$ and $O(D\log{n})$ deterministic leader election algorithm for general graphs.

\section{A Randomized Algorithm} \label{sec:random}

In this section, we present a simple randomized Monte Carlo algorithm that works in a constant number of rounds. Algorithm~\ref{alg:main} is entirely local, as nodes do not require any knowledge of $n$. Nevertheless, we show that we can sub-sample a small number of candidates (using only local knowledge) that then attempt to become leader. In the remainder of this section, we prove the following result. 

\begin{theorem} \label{thm:randomized}
  There exists a Monte Carlo randomized leader election algorithm that, with high probability, succeeds in $n$-node networks of diameter at most two in $O(1)$ rounds, while sending $O(n\log^3n)$ messages.
\end{theorem}

\begin{algorithm}[h]
\begin{algorithmic}[1]
	\State Each node $v \in V$ selects itself to be a ``candidate'' with probability $\frac{1 + \log{(d_v)}}{d_v}$, where $d_v$ is the degree of $v$. 
	\State \textbf{if} {$v$ becomes a candidate} \textbf{then} $v$ sends its ID to all its neighbors.
	\State Each node acts as a ``referee node'' for all its candidate neighbors (including, possibly itself).
	\State If a node $w$ receives ID's from its neighbors $v_1, v_2, \ldots, v_j$ (say), then $w$ computes the minimum ID of those and sends it back to those neighbors. That is, $w$ sends $\text{min}\left\{ID(v_1), 
	ID(v_2), \ldots, ID(v_j)\right\}$ back to each of $v_1, v_2, \ldots, v_j$.
	\State A node $v$ decides that it is the leader if and only if it receives its own ID from \emph{all} its neighbors. Otherwise $v$ decides that it is not the leader.
\end{algorithmic}
\caption{Randomized leader election in $O(1)$ rounds and $O(n\log^3{n})$ message complexity}
\label{alg:main}
\end{algorithm}


\subsection{Proof of Correctness: Analyzing the number of candidates selected}\label{sec:number-of-candidates}
We use the following property of diameter-$2$ graphs crucially in our algorithm.

\begin{obs}\label{main-diameter-two-property}
	Let $G = (V, E)$ be a graph of diameter $2$. Then for any $u, v \in V$, either $(u, v) \in E$ or $\exists w \in V$ such that $(u, w) \in E$ and $(v, w) \in E$, i.e., $u$ and $v$ have at least one common neighbor $w$ (say).
\end{obs}

We note that if one or more candidates are selected, then only the candidate node with the minimum ID is selected as the leader. That is, the leader is unique, and therefore the algorithm produces the correct output. The only case when the algorithm may be wrong is if no candidates are selected to begin with, in which case no leader is selected. In this section, we show that, with high probability, at least two candidates are selected.

We make use of the following fact in order to show that.

\begin{lemma}\label{lemma-Lagrangian-Hessian}
	Let $f(x_1, x_2, \ldots, x_n)$ be a function of $n$ variables $x_1, x_2, \ldots, x_n$, where $x_1, x_2, \ldots, x_n$ are positive reals. $f$ is defined as
	\begin{align*}
		&f(x_1, x_2, \ldots, x_n)  \defeq  \sum_{i = 1}^n \frac{1 + \log{x_i}}{x_i}\text{.}
	\end{align*}
	Let $C$ be a constant $\geq n\sqrt{2}$. Then $f(x_1, x_2, \ldots, x_n)$ is minimized, subject to the constraint $\sum_{i = 1}^n x_i = C$, when $x_i  =  \frac{C}{n}$, for all $1 \leq i \leq n$. The minimum value that 
	$f(x_1, x_2, \ldots, x_n)$ takes is at the point
	\begin{center}
		$(\frac{C}{n}, \frac{C}{n}, \ldots, \frac{C}{n})$, and is given by
	\end{center}	
	\begin{align*}
		f^{\text{min}}   =   f(\frac{C}{n}, \frac{C}{n}, \ldots, \frac{C}{n})   =   \frac{n^2}{C}(1 + \log{(\frac{C}{n})})\text{.}
	\end{align*}
\end{lemma}
\begin{proof}
	We use standard Lagrangian optimization techniques to show this. Please refer to the appendix for the full proof.
\end{proof}

\begin{lemma}\label{lemma-expected-number-of-selected-candidates}
	Let $X$ be a random variable that denotes the total number of candidates selected in Algorithm~\ref{alg:main}. Then the expected number of selected candidates is lower-bounded by
	\begin{center}
		$E[X]  >  2 + \frac{1}{2}\log{n}$.
	\end{center}
\end{lemma}
\begin{proof}
	Let $X_v$ be an indicator random variable that takes the value $1$ if and only if $v$ becomes a candidate. Then $E[X_v] = \Pr[X_v = 1] = \frac{1 + \log{d_v}}{d_v}$. Thus if $X$ denotes the total number of candidates selected, then
    \begin{align*}
        &E[X]  =  \sum_{v \in V} E[X_v]   =  \sum_{v \in V} \frac{1 + \log{(d_v)}}{d_v}\text{.}
    \end{align*}
    
	Since $G$ is connected, $m \geq n-1 \implies 2m \geq 2n-2 > n\sqrt{2}$, i.e., the precondition for the applicability of Lemma \ref{lemma-Lagrangian-Hessian} is satisfied.
	
	Thus by By Lemma \ref{lemma-Lagrangian-Hessian}, $E[X]$ is minimized subject to the constraint 
	$$\sum_{v \in V(G)}d_v  =  2m\text{,}$$
	when $d_v = \frac{2m}{n}$ for all $v \in V(G)$, i.e., when $G$ is regular.\\
	
	\textbf{Case 1 ($n-1  \leq  m  \leq  n^{\frac{3}{2}}$):} The minimum value that $E[X]$ takes is given by
	\begin{align*}
        &\left.E[X]\right\rvert_{\text{min}}   =   \frac{n^2}{2m}(1 + \log{(\frac{2m}{n})})\\
        &>   \frac{n^2}{2m} \tag{since $1 + \log{(\frac{2m}{n})}  >  1$}\\
        &\geq   \frac{\sqrt{n}}{2} \tag{since $m \leq n^{\frac{3}{2}}$}
	\end{align*}
	
	\textbf{Case 2 ($n^{\frac{3}{2}}  <  m  \leq  {n \choose 2}$):} The minimum value that $E[X]$ takes is given by
	\begin{align*}
		&E[X]\rvert_{\text{min}}   =   \frac{n^2}{2m}(1 + \log{(\frac{2m}{n})})\\
        &>   1 + \log{(\frac{2n^{\frac{3}{2}}}{n})} \tag{since $\frac{n^2}{2m} > 1$ and $m > n^{\frac{3}{2}}$}\\
		&=   1 + \log{2} + \log{(n^{\frac{1}{2}})}   =   2 + \frac{1}{2}\log{n}\text{.}
	\end{align*}
\end{proof}

We use the following variant of Chernoff Bound \cite{Mitzenmacher_2017_Book} to show \emph{concentration}, i.e., to show that the number of candidates selected is not too less than its expected value.
\begin{theorem}[Chernoff Bound]\label{theorem-Chernoff-Bound-2}
	Let $X_1, X_2, \ldots, X_n$ be independent indicator random variables, and let $X  =  \sum_{i=1}^n X_i$. Then the following Chernoff bound holds: for $0 < \delta < 1$,
	\begin{center}
		$\Pr[X  \leq  (1 - \delta)\mu]   \leq   (\frac{e^{-\delta}}{(1 - \delta)^{1 - \delta}})^{\mu}$, where $\mu  \defeq  E[X]$.
	\end{center}
\end{theorem}

\begin{lemma}\label{lemma-number-of-selected-candidates-is-at-least-2}
	If $X$ denotes the number of candidates selected, then $\Pr[X \leq 1]   <   n^{-\frac{1}{3}}$.
\end{lemma}
\begin{proof}
	We set $\delta  =  \frac{2 + \log{n}}{4 + \log{n}}$. Then clearly $0 < \delta < 1$, and $1 - \delta  =  \frac{2}{4 + \log{n}}$. Again, from Lemma \ref{lemma-expected-number-of-selected-candidates}, we have 
	that
	\begin{align*}
		&\mu   =   E[X]   >   2 + \frac{1}{2}\log{n}     \implies     (1 - \delta)\mu   >   (1 - \delta)(2 + \frac{1}{2}\log{n})\\
		&=   \frac{2}{4 + \log{n}} . (2 + \frac{1}{2}\log{n})   =   1\text{.}
	\end{align*}
	
	Then by Theorem \ref{theorem-Chernoff-Bound-2}, $\Pr[X \leq 1]   \leq   \Pr[X  \leq  (1 - \delta)\mu]   \leq   (\frac{e^{-\delta}}{(1 - \delta)^{1 - \delta}})^\mu$. Now
	\begin{align*}
		&\frac{e^{-\delta}}{(1 - \delta)^{1 - \delta}}   =   \frac{e^{-\frac{2 + \log{n}}{4 + \log{n}}}}{(\frac{2}{4 + \log{n}})^{\frac{2}{4 + \log{n}}}}   
		=   (\frac{e^{-(2 + \log{n})}}{(\frac{2}{4 + \log{n}})^2})^{\frac{1}{4 + \log{n}}}\\
		&=   (\frac{(2 + \frac{1}{2}\log{n})^2}{e^{2 + \log{n}}})^{\frac{1}{4 + \log{n}}}   <   (\frac{(2 + \frac{1}{2}\log{n})^2}{ne^2})^{\frac{1}{4 + \log{n}}}\\
		&\implies     (\frac{e^{-\delta}}{(1 - \delta)^{1 - \delta}})^{\mu}   <   (\frac{(2 + \frac{1}{2}\log{n})^2}{ne^2})^{\frac{\mu}{4 + \log{n}}}\\
		&<   (\frac{(2 + \frac{1}{2}\log{n})^2}{ne^2})^{\frac{2 + \frac{1}{2}\log{n}}{4 + \log{n}}} \tag{since $\mu   >   2 + \frac{1}{2}\log{n}$ from Lemma \ref{lemma-expected-number-of-selected-candidates}}\\
		&=   (\frac{(2 + \frac{1}{2}\log{n})^2}{ne^2})^{\frac{1}{2}}   =   \frac{2 + \frac{1}{2}\log{n}}{e\sqrt{n}}\text{.}
	\end{align*}
	
	Hence, $\Pr[X \leq 1]   \leq   (\frac{e^{-\delta}}{(1 - \delta)^{1 - \delta}})^\mu   <   \frac{2 + \frac{1}{2}\log{n}}{e\sqrt{n}}   <   \frac{1}{n^{\frac{1}{3}}}$, assuming $\frac{2 + \frac{1}{2}\log{n}}{e}   <   
	n^{\frac{1}{6}}$, which is asymptotically true.
\end{proof}
\subsection{Computing the message complexity} \label{sec:analysis-message-complexity}

Note that the expected total message complexity of the algorithm can be bounded as follows. Let $M^{\text{entire}}$ be a random variable that denotes the total messages sent during the course of the algorithm. Let $M_v$ be the number of messages sent by node $v$. Thus

$$M^{\text{entire}} = \sum_{v \in V} M_v\text{.}$$

A node $v$ becomes a candidate with probability $\frac{1 + \log{d_v}}{d_v}$ and, if it does, it sends $d_v$ messages (the referee nodes reply to these, but that increases the total number of messages by only a factor of 2). Hence by linearity of expectation,  it follows that
\begin{align*}
    &E[M^{\text{entire}}]   =   \sum_{v \in V} E[M_v]   =   \sum_{v \in V}2\frac{1 + \log{d_v}}{d_v}d_v\\
    &=   2\sum_{v \in V}  (1+\log{d_v})   \leq   2\sum_{v \in V} (1 + \log n)\\
    &\leq   2n  +  2n\log{n}\text{.}
\end{align*}

To show concentration, we cannot directly apply a standard Chernoff bound, since that works for $0$-$1$ random variables only, whereas the $M_v$s are not $0$-$1$ random variables (they take values either $0$ or $d_v$). To handle this, we bucket the degrees into (at most) $\log{n}$ categories based on their value then use a Chernoff bound as detailed below.

We use the following variant of Chernoff Bound \cite{Mitzenmacher_2017_Book} in the following analysis.
\begin{theorem}[Chernoff Bound]\label{theorem-Chernoff-Bound-1}
	Let $X_1, X_2, \ldots, X_n$ be independent indicator random variables, and let $X  =  \sum_{i=1}^n X_i$. Then the following Chernoff bound holds: for $R \geq 6E[X]$, $\Pr[X  \geq  R]   \leq   2^{-R}$.
\end{theorem}

\begin{definition}
	Let $k$ be a positive integer such that $2^{k-1}  <  n  \leq  2^k$. For $0 \leq j \leq k$, let $V_j \subset V$ be the set of vertices with degree in $(2^{j-1}, 2^j]$, i.e., if $v \in V_j$, then $2^{j-1}  < d_v  \leq 2^j$.
	Thus
	\begin{align*}
		&V_0  \defeq  \left\{v \in V\ |\ d_v = 1\right\}\text{,}\\
		&V_1  \defeq  \left\{v \in V\ |\ d_v = 2\right\}\text{,}\\
		&V_2  \defeq  \left\{v \in V\ |\ d_v = 3\text{ or }d_v = 4\right\}\text{,}\\
		&V_3  \defeq  \left\{v \in V\ |\ d_v \in \left\{5, 6, 7, 8\right\}\right\}\text{, and so on.}
	\end{align*}
\end{definition}

\begin{remark}
	We note that
	$$\sum_{j = 0}^k n_j  =  n\text{, where }n_j = |V_j|\text{ for }0 \leq j \leq k\text{.}$$
	In particular, $n_j  \leq  n$ for all $j \in [0, k]$.
\end{remark}


\paragraph{Counting the number of messages sent in the first round by each individual node}
\begin{enumerate}

    \item \textbf{Analyzing vertices with degree $\leq 2$:} We recall that $X_v$ is an indicator random variable that takes the value $1$ if and only if $v$ becomes a candidate. Then $\Pr[X_v = 1] = 1$ if $v \in V_0 \cup V_1$, i.e., every vertex with degree $1$ or degree $2$ selects itself to be a candidate, deterministically.\\

    For $v \in V$, let $m_v$ denote the number of messages that $v$ sends. So $m_v = d_v$ if $v$ becomes a candidate, and $m_v = 0$ otherwise. Let $M_j$ be the total number of messages that members of $V_j$ send, i.e.,
    \begin{align*}
        &M_j   \defeq  \sum_{v \in V_j} m_v   \leq   \sum_{v \in V_j} d_v\\
        &\leq   \sum_{v \in V_j} 2^j   =   n_j . 2^j   \leq   n . 2^j\\
        &\implies   M_0 \leq n\text{ and }M_1 \leq 2n\text{.}
    \end{align*}

    \item \textbf{Analyzing vertices with degree $> 2$:} We recall that for $v \in V$, $X_v$ is an indicator random variable that takes the value $1$ if and only if $v$ becomes a candidate. Let $i$ be an integer in $[2, k]$ and let $v \in V_i$. Then 
    \begin{obs}
        $\frac{i}{2^i}   <   E[X_v]   <   \frac{3i}{2^i}$.
    \end{obs}
    \begin{proof}
        For $v \in V_i$, $2^{i-1} < d_v \leq 2^i$. So
        \begin{align*}
            &E[X_v]  =  \Pr[X_v = 1] \tag{since $X_v$ is an indicator random variable}\\
            &=   \frac{1 + \log{d_v}}{d_v}\\
            &\implies   \frac{1 + \log{(2^{i-1})}}{2^i}   <   E[X_v]   <   \frac{1 + \log{(2^i)}}{2^{i-1}} \tag{since $2^{i-1}  <  d_v  \leq  2^i$}\\
            &\text{or, }  \frac{i}{2^i}   <   E[X_v]   <   \frac{i+1}{2^{i-1}}   \leq   \frac{3i}{2^i} \tag{since $i \geq 2  \implies  \frac{3i}{2}  \geq  i+1$}
        \end{align*}
    \end{proof}

    For $0 \leq j \leq k$, let $Y_j$ be a random variable that denotes the total number of candidates selected from $V_j$.
    \begin{obs}
        For $2 \leq i \leq k$, $\frac{in_i}{2^i}   <   E[Y_i]   <   \frac{3in_i}{2^i}$.
    \end{obs}
    \begin{proof}
        \begin{align*}
            &Y_i  =  \sum_{v \in V_i} X_v   \implies   E[Y_i]  =  E[\sum_{v \in V_i} X_v]\\
            &=  \sum_{v \in V_i} E[X_v] \tag{by linearity of expectation}\\
            &\implies   \sum_{v \in V_i} \frac{i}{2^i}   <   E[Y_i]   <   \sum_{v \in V_i} \frac{3i}{2^i}\\
            &\implies   \frac{in_i}{2^i}   <   E[Y_i]   <   \frac{3in_i}{2^i}\text{.}
        \end{align*}
    \end{proof}

    \begin{remark}
        $\forall u, v \in V(G)$, $u \neq v$, $X_u$ and $X_v$ are independent, and for $0 \leq j \leq k$, we define $Y_j$ as
        \begin{align*}
            Y_j  =  \sum_{v \in V_j} X_v\text{,}
        \end{align*}
        i.e., $Y_j$ is a sum of independent, $0$-$1$ random variables. Hence we can use Theorem \ref{theorem-Chernoff-Bound-1} to show that $Y_j$ is concentrated around its mean (expectation).
    \end{remark}

    We recall that for $0 \leq j \leq k$, $M_j$ is the total number of messages that members of $V_j$ send, i.e., for $2 \leq i \leq k$,
    \begin{align*}
        &M_i  =  \sum_{v \in V_i} m_v   =   \sum_{v \in V_i, X_v = 1} d_v\text{.}
    \end{align*}

    \begin{lemma}
        For any integer $i \in [2, k]$, it holds that
    	\begin{center}
            $\Pr[M_i  \geq  24n\log^2{n}]   \leq   \frac{1}{n^4}$.
        \end{center}
    \end{lemma}
    \begin{proof}
        \begin{align*}
            &M_i  =  \sum_{v \in V_i} m_v   =   \sum_{v \in V_i, X_v = 1} d_v\\
            &\implies     \sum_{v \in V_i, X_v = 1} 2^{i-1}   <   M_i   \leq   \sum_{v \in V_i, X_v = 1} 2^i \tag{since $2^{i-1} < d_v  \leq  2^i$}\\
            &\implies   2^{i-1} . Y_i   <   M_i   \leq   2^i . Y_i\text{.}
        \end{align*}

        \textbf{Case $1$ ($E[Y_i] = 0$):} $E[Y_i] = 0$ if and only if $n_i = 0$, i.e., if and only if $\nexists v \in V$ such that $2^{i-1}  <  d_v  \leq  2^i$. But $n_i = 0   \implies   V_i = \phi$, the empty set. Therefore, $M_i = 0$.\\

        \textbf{Case $2$ ($0  <  E[Y_i]  <  1$):} Assuming $n \geq 3$, $4\log{n}  >  6  >  6E[Y_i]$. Therefore, by Theorem \ref{theorem-Chernoff-Bound-1},
        \begin{align*}
            &\Pr[Y_i \geq 4\log{n}]  \leq  2^{-4\log{n}}  =  n^{-4}\\
            &\implies   \Pr[M_i  \geq  2^i . 4\log{n}]   \leq   n^{-4} \tag{since $M_i  \leq  2^i . Y_i$}\\
            &\implies   \Pr[M_i  \geq  8n\log{n}]   \leq   n^{-4} \tag{since $i  \leq  k  <  \log{n} + 1$}
        \end{align*}

        \textbf{Case $3$ ($E[Y_i] \geq 1$):} We have shown before that $E[Y_i]  \leq  \frac{3in_i}{2^i}$. But $n_i \leq n$ for all $2 \leq i \leq k$. Hence $E[Y_i]  \leq  \frac{3ni}{2^i}$. Assuming $n \geq 3$, $4\log{n}  >  6$. Therefore, by Theorem \ref{theorem-Chernoff-Bound-1},
        \begin{align*}
            &\Pr[Y_i  \geq  12n\log{n} . \frac{i}{2^i}]   \leq   \Pr[Y_i \geq 4\log{n}E[Y_i]]\\
            &\leq  2^{-4\log{n}E[Y_i]}  =  n^{-4E[Y_i]}  \leq  n^{-4} \tag{since $E[Y_i] \geq 1$}\\
            &\implies   \Pr[M_i  \geq  12in\log{n}]   \leq   n^{-4} \tag{since $M_i  \leq  2^i . Y_i$}
        \end{align*}
        
        But we have, $i  \leq  k  <  \log{n} + 1  <  2\log{n}   \implies   12in\log{n}  <  24n\log^2{n}$. Hence
        \begin{center}
            $\Pr[M_i  \geq  24n\log^2{n}]   \leq   \Pr[M_i  \geq  12in\log{n}]   \leq   n^{-4}$.
        \end{center}
    \end{proof}

\end{enumerate}


\paragraph{Combining the effects of all the nodes}
\begin{lemma}
    If $M$ denotes the total number of messages sent by the candidates (in the first round only), then $\Pr[M  \geq  27n\log^3{n}]   <   \frac{1}{n^3}$.
\end{lemma}
\begin{proof}
    \begin{align*}
        &M   \defeq   \sum_{i = 0}^k M_i   =   M_0 + M_1 + \sum_{i = 2}^k M_i\\
        &\leq   n + 2n + \sum_{i = 2}^k M_i \tag{since $M_0 \leq n$ and $M_1 \leq 2n$}\\
        &=   3n + \sum_{i = 2}^k M_i\text{.}
    \end{align*}
    
    But for $2 \leq i \leq k$, $\Pr[M_i  \geq  24n\log^2{n}]   \leq   \frac{1}{n^4}$. Taking the union bound over $2 \leq i \leq k$,
    \begin{align*}
        &\Pr[M_{i'}  \geq  24n\log^2{n}]\text{ for some }i' \in [2, k]\text{ is }\leq   \frac{\log{n}}{n^4}   <   \frac{1}{n^3}\\
        &\implies   \Pr[\sum_{i = 2}^k M_i   \geq   24n\log^3{n}]   <   \frac{1}{n^3}\\
        &\implies   \Pr[3n   +   \sum_{i = 2}^k M_i   \geq   3n + 24n\log^3{n}]     <     \frac{1}{n^3}\\
        &\implies   \Pr[M  \geq  3n + 24n\log^3{n}]   <   \frac{1}{n^3} \tag{since $M   \leq   3n   +   \sum_{i = 2}^k M_i$}
    \end{align*}
    But $3n  \leq  3n\log^3{n}$ for $n \geq 2$, or, $3n + 24n\log^3{n}   \leq   27n\log^3{n}$. Hence
    \begin{align*}
        &\Pr[M  \geq  27n\log^3{n}]   \leq   \Pr[M  \geq  3n + 24n\log^3{n}]   <   \frac{1}{n^3}\text{.}
    \end{align*}
\end{proof}

\begin{lemma}
    If $M^{\text{entire}}$ denotes the total number of messages sent during the entire run of Algorithm~\ref{alg:main}, then $\Pr[M^{\text{entire}}  \geq  54n\log^3{n}]   <   \frac{1}{n^3}$.
\end{lemma}
\begin{proof}
    Let $M'$ denote the number of messages sent by the ``referee'' nodes in the \emph{second} round of the algorithm. We recall that $M$ is the number of messages sent by the ``candidate'' nodes in the \emph{first} round of the algorithm. Then $M' \leq M$, and $M^{\text{entire}}  =  M + M'  \leq  2M$, and the result follows.
\end{proof}

This completes the proof of Theorem~\ref{thm:randomized}.
\section{A Lower Bound for Randomized Algorithms} \label{sec:lb}

In this section we show that $\Omega(n)$ is a lower bound on the message complexity for solving leader election with any randomized algorithm in diameter-two networks. Notice that \cite{Kutten_2015_TCS} show a lower bound of $\Omega(\sqrt{n})$ for the special case of diameter $1$ networks, and we know from \cite{Kutten_2015_JACM} that, for the message complexity becomes $\Omega(m)$ for (most) diameter $3$ networks. Thus, Theorem~\ref{thm:msgRandomized} completes the picture regarding the message complexity of leader election when considering networks according to their diameter.

\begin{theorem} \label{thm:msgRandomized}
    Any algorithm that solves implicit leader election with probability at least $\frac{1}{2} + \epsilon$ in any $n$-node network with diameter at most $2$, for any constant $\epsilon > 0$, sends at least $\Omega(n)$ messages in expectation. This holds even if nodes have unique IDs and know the network size $n$.
\end{theorem}

In the remainder of this section, we prove Theorem~\ref{thm:msgRandomized}.

Assume towards a contradiction, that there is an algorithm that elects a leader with probability $\frac{1}{2}+\epsilon$ that sends $o(n)$ messages with probability approaching $1$. In other words, we assume that the event where the algorithm sends more than $o(n)$ messages (of arbitrary size) happens with probability at most $o(1)$.

\paragraph{Unique IDs vs.\ Anonymous}
Before describing our lower bound construction, we briefly recall a simple reduction used in \cite{Kutten_2015_TCS} that shows that assuming unique IDs does not change the success probability of the algorithm by more than $\frac{1}{n}$: Since we assume that nodes have knowledge of $n$, it is straightforward to see that nodes can obtain unique IDs (whp) by choosing a random integer in the range $[1,n^c]$, for some constant $c\ge 4$. Thus, we can simulate an algorithm that requires unique IDs in the anonymous case and the simulation will be correct with high probability. Suppose that there is an algorithm $A$ that can break the message complexity bound of Theorem~\ref{thm:msgRandomized} while succeeding with probability $\ge \frac{1}{2} + \epsilon$, for some constant $\epsilon>0$, when nodes have unique IDs. Then, the above simulation yields an algorithm $A'$ that works in the case where nodes are \emph{anonymous} with the same message complexity bound as algorithm $A$ and succeeds with probability at least $(\frac{1}{2}  +  \epsilon - \frac{1}{n})   \geq   \frac{1}{2} + \epsilon'$, for some constant $\epsilon'  >  0$. We conclude that proving the lower bound for the anonymous case is sufficient to imply a lower bound for the case where nodes have unique IDs.

\paragraph{The Lower Bound Graph}
Our lower bound is inspired by the bridge crossing argument of \cite{Kutten_2015_JACM} and \cite{Pai_2017}. For simplicity, we assume that $\frac{n}{4}$ is an integer. Consider two cliques $C_1$ and $C_2$ of $\frac{n}{2}$ nodes each and let $G'$ be the $n$-node graph consisting of the two (disjoint) cliques. The \emph{port numbering} of an edge $e  =  (u_i,v_j)  \in  E(G')$ refers to the port number at $u_i$ and the respective port number at $v_j$ that connects $e$. The port numberings of the edges defines an \emph{instance of $G'$}.

Given an instance of $G'$, we will now describe how to obtain an instance of graph $G$ that has the same node set as $G'$. Fix some arbitrary enumeration $u_1, \dots, u_{\frac{n}{2}}$ of the nodes \footnote{This enumeration is used solely for the description of the lower bound construction and is unbeknownst to the nodes.} in $C_1$ and similarly let $v_1,\dots,v_{\frac{n}{2}}$ be an enumeration of the nodes in $C_2$. To define the edges of $G$, we randomly choose a maximal matching $M_1$ of $\frac{n}{4}$ edges in the subgraph $C_1$. Consider the set of edges $M_2' = \{ (v_i,v_j) \mid \exists (u_i,u_j) \in M_1 \}$, which is simply the matching in $C_2$ corresponding to $M_1$ in $C_1$. We define $M_2$ to be a randomly chosen maximal matching on $C_2$ when using only edges in $E(G') \setminus M_2'$. Then, we remove all edges in $M_1 \cup M_2$ from $G'$. So far, we have obtained a graph where each node has one \emph{unwired port}.

The edge set of $G$ consists of all the remaining edges of $G'$ in addition to the set $M = \{ (u_1,v_1),\dots,(u_{\frac{n}{2}},v_{\frac{n}{2}})\}$, where we connect these \emph{bridge edges} by using the unwired ports that we obtained by removing the edges as described above. We say that an edge is an \emph{intra-clique edge} if it has both endpoints in either $C_1$ or $C_2$. Observe that the intra-clique edges of $G$ are a subset of the intra-clique edges of $G'$. Figure~\ref{fig:messagelb} gives an illustration of this construction.

\begin{lemma} \label{lem:lbgraph}
Graph $G$ is an $n$-node network of diameter $2$ and the port numbering of each intra-clique edge in $G$ is the same as of the corresponding edge in $G'$.
\end{lemma}
\begin{proof}
By construction, each node in $C_1$ has the same port numbering in both graphs, except for its (single) incident edge that was replaced with a {bridge edge} to some node in $C_2$, thus we focus on showing that $G$ has diameter $2$.

We will show that node $u_i \in C_1$ has a path of length $2$ to every other node. Observe that any two nodes $u_i, u_j \in C_1$ both have $\frac{n}{2} - 2$ incident intra-clique edges and since $\frac{n}{2} - 2   >   \frac{|C_1|}{2}$ they must both have a common neighbor. Now, consider some node $v_j \in C_2$ and assume that $j \neq i$, as otherwise there is the matching edge $(u_i,v_i) \in M$. If $(u_i,u_j) \in E(G)$, then again the result follows because $(u_j,v_j) \in M$. Otherwise, there must be the path $u_i\rightarrow v_i\rightarrow v_j$, since, by construction, the edge $(v_i,v_j) \in M_2'$ and hence $(v_i,v_j) \notin M_2$.

A symmetric argument shows that every node has distance $\le 2$ from a given node in $C_2$.
\end{proof}

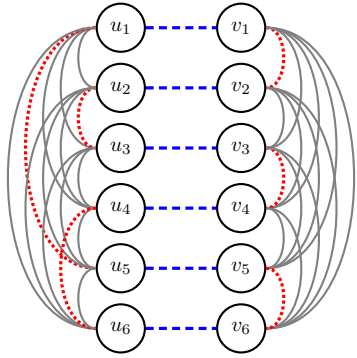
\begin{figure}[t]
\begin{center}
\pgfdeclarelayer{background}
\pgfdeclarelayer{inbetween}
\pgfsetlayers{background,inbetween,main}
\begin{tikzpicture}[node distance=1.0cm,scale=0.8, every node/.style={scale=0.8}]
\tikzstyle{v}=[circle,draw=black,fill=white!30,thick,inner sep=2pt,minimum size=8mm]
\tikzstyle{link}=[-,black!50,thick,auto]
\tikzstyle{llink}=[-,black!50,thick,auto,in=180,out=180]
\tikzstyle{rlink}=[-,black!50,thick,auto,in=0,out=0]
\tikzstyle{arc}=[draw=black,rectangle,rounded corners,fill=lightgray,font=\normalsize]
\tikzstyle{arcNeighborhoodZ}=[semitransparent,thick,draw=none,rectangle,rounded corners,fill=blue!20,pattern=horizontal lines light blue]
\tikzstyle{arcNeighborhoodZ1}=[arcNeighborhoodZ,fill=red!20,dashed]
\node[v]  (u001)                  {$u_{1}$};
\node[v]  (u002)  [below of=u001] {$u_{2}$};
\node[v]  (u003)  [below of=u002] {$u_{3}$};
\node[v]  (u004)  [below of=u003] {$u_{4}$};
\node[v]  (u005) [below of=u004] {$u_{5}$};
\node[v]  (u006) [below of=u005] {$u_{6}$};

\node[v]  (v001)  [xshift=2cm]    {$v_{1}$};
\node[v]  (v002)  [below of=v001] {$v_{2}$};
\node[v]  (v003)  [below of=v002] {$v_{3}$};
\node[v]  (v004)  [below of=v003] {$v_{4}$};
\node[v]  (v005) [below of=v004] {$v_{5}$};
\node[v]  (v006) [below of=v005] {$v_{6}$};

\draw[llink] (u001.west) to (u002.west);
\draw[llink] (u001.west) to (u003.west);
\draw[llink] (u001.west) to (u004.west);
\draw[llink,very thick,densely dotted,red] (u001.west) to (u005.west);
\draw[llink] (u001.west) to (u006.west);
\draw[llink,very thick,densely dotted,red] (u002.west) to (u003.west);
\draw[llink] (u002.west) to (u004.west);
\draw[llink] (u002.west) to (u005.west);
\draw[llink] (u002.west) to (u006.west);
\draw[llink] (u003.west) to (u004.west);
\draw[llink] (u003.west) to (u005.west);
\draw[llink] (u003.west) to (u006.west);
\draw[llink] (u004.west) to (u005.west);
\draw[llink,very thick,densely dotted,red] (u004.west) to (u006.west);
\draw[llink] (u005.west) to (u006.west);

\draw[rlink,very thick,densely dotted,red] (v001.east) to (v002.east);
\draw[rlink] (v001.east) to (v003.east);
\draw[rlink] (v001.east) to (v004.east);
\draw[rlink] (v001.east) to (v005.east);
\draw[rlink] (v001.east) to (v006.east);
\draw[rlink] (v002.east) to (v003.east);
\draw[rlink] (v002.east) to (v004.east);
\draw[rlink] (v002.east) to (v005.east);
\draw[rlink] (v002.east) to (v006.east);
\draw[rlink,very thick,densely dotted,red] (v003.east) to (v004.east);
\draw[rlink] (v003.east) to (v005.east);
\draw[rlink] (v003.east) to (v006.east);
\draw[rlink] (v004.east) to (v006.east);
\draw[rlink] (v004.east) to (v005.east);
\draw[rlink,very thick,densely dotted,red] (v005.east) to (v006.east);

\draw[link,blue,very thick,densely dashed] (u001) to (v001);
\draw[link,blue,very thick,densely dashed] (u002) to (v002);
\draw[link,blue,very thick,densely dashed] (u003) to (v003);
\draw[link,blue,very thick,densely dashed] (u004) to (v004);
\draw[link,blue,very thick,densely dashed] (u005) to (v005);
\draw[link,blue,very thick,densely dashed] (u006) to (v006);

\begin{pgfonlayer}{background}
\end{pgfonlayer}
\begin{pgfonlayer}{inbetween}
\end{pgfonlayer}
\end{tikzpicture}
\end{center}
\caption{\boldmath The lower bound graph construction used in Theorem~\ref{thm:msgRandomized} for $n=12$, with cliques $C_1$ and $C_2$, where $V(C_1) = \{u_1,\dots,u_6\}$ and $V(C_2) = \{v_1,\dots,v_6\}$. The dotted red edges are the edges in $M_1$ and $M_2$ that are removed from $C_1$ and $C_2$ when constructing $G$ and the blue dashed inter-clique edges are given by the maximal matching $M$ between $C_1$ and $C_2$. Each blue edge incident to some node $u_i$ is connected by using the port number of $u_i$'s (removed) red edge.}
\label{fig:messagelb}
\end{figure}
 

A \emph{state} $\sigma$ of the nodes in $C_1$ is a $\frac{n}{2}$-size vector of the local states of the $\frac{n}{2}$ nodes in $C_1$. Since we assume that nodes are anonymous, a state $\sigma$ that is reached by the nodes in $C_1$, can also be reached by the nodes in $C_2$. More formally, when executing the algorithm on the disconnected network $G'$, we can observe that every possible state $\sigma$ (of $\frac{n}{2}$ nodes) has the same probability of occurring in $C_1$ as in $C_2$. Thus, a state where there is exactly one leader among the $\frac{n}{2}$ nodes of a clique in $G'$, is reached with some specific probability $q$ depending on the algorithm. By a slight abuse of notation, we also use $G'$ and $G$ to denote the event that the algorithm executes on $G'$ respectively $G$. For the probability of the event $\One$, which occurs when there is exactly $1$ leader among the $n$ nodes, we get
\begin{align} \label{eq:maxProb}
    \Prob{ \One\ \middle|\ G' } = 2q(1 - q) \le \frac{1}{2}\text{,}
\end{align}
which holds for any value of $q$. Since $G'$ is disconnected, the algorithm does not need to succeed with nonzero probability when being executed on $G'$. However, below we will use this observation to obtain an upper bound on the probability of obtaining (exactly) one leader in $G$.

Now consider the execution on the diameter $2$ network $G$ (obtained by modifying the ports of $G'$ as described above) and let $\disjoint$ be the event that no message is sent across the bridges between $C_1$ and $C_2$. Since we assume the port numbering model where nodes are unaware of their neighbors initially, it follows by Lemma~\ref{lem:lbgraph} that
\begin{align}
    \Prob{ \One\ \middle|\ \disjoint, G } = \Prob{ \One\ \middle|\ G' }\text{.} \label{eq:same}
\end{align}

Let $M$ be the event that the algorithm sends $o(n)$ messages. Recall that we assume towards a contradiction that $\Prob{M \mid G} = 1 - o(1)$.
\begin{lemma} \label{lem:disjoint}
    $\Prob{ C_1 \leftrightarrow C_2\ \middle|\ G, M } = o(1)$.
\end{lemma}
\begin{proof}
The proof is inspired by the guessing game approach of \cite{Gilbert_2018} and Lemma~16 in \cite{Pai_2017}. Initially, any node $u \in C_1$ has $\frac{n}{2} - 1$ ports that are all equally likely (i.e., a probability $p = \frac{1}{\frac{n}{2} - 1}$) to be connected to the (single) bridge edge incident to $u$. As $u$ sends messages to other nodes, it might learn about some of its ports connecting to non-bridge edges and hence this probability can increase over time. However, we condition on event $M$, i.e., the algorithm sends at most $o(n)$ messages in total and hence at least $\frac{n}{4}$ ports of each node $u$ remain unused at any point.

It follows that the probability of some node $u$ to activate a (previously unused) port that connects a bridge edge is at most $\frac{4}{n}$ at any point of the execution. Let $X$ be the total number of ports connecting bridge edges that are activated during the run of the algorithm and let $X_u$ be the indicator random variable that is $1$ iff node $u$ sends a message across its bridge edge. Let $S_u$ be the number of messages sent by node $u$. It follows by the hypergeometric distribution that 
\[
  \textrm{E}[X_u \mid G, M ] = S_u \frac{1}{\Theta(n)}, 
\]
for each node $u$ and hence, 
\begin{center}
    $\textrm{E}[X \mid G, M]   =   \sum_{u \in V(G)}\frac{S_u}{\Theta(n)}   =   \frac{1}{\Theta(n)}\sum_{u \in V(G)}S_u   =   o(1)$,
\end{center}
where we have used the fact that $\sum_{u \in V(G)}  S_u = o(n)$ due to conditioning on event $M$. By Markov's Inequality, it follows that the event $C_1 \leftrightarrow C_2$, i.e., $X\ge 1$, occurs with probability at most $o(1)$.
\end{proof}

We now combine the above observations to obtain
\begin{align}
    \Prob{ \One\ \middle|\ G, M } 
    &=   \Prob{ \One\ \middle|\ \disjoint, G, M }\Prob{\disjoint\ \middle|\ G, M} \notag\\
    &+   \Prob{ \One\ \middle|\ C_1 \leftrightarrow C_2, G, M } \Prob{ C_1 \leftrightarrow C_2\ \middle|\ G, M}\notag \\
    &\le \Prob{ \One\ \middle|\ \disjoint, G, M } + o(1) \tag{by Lem.~\ref{lem:disjoint}} \\
    &\le \frac{1}{2} + o(1),  \label{eq:uppBnd1}
\end{align}
where the last inequality follows by first using \eqref{eq:same} and noting that the upper bound~\eqref{eq:maxProb} still holds when conditioning on the event $M$.

Finally, we recall that the algorithm succeeds with probability at least $\frac{1}{2} + \epsilon$ and $\Prob{M \mid G} \ge 1 - o(1)$, which yields
\begin{align*}
    \frac{1}{2} + \epsilon 
    \le \Prob{ \One\ \middle|\ G }
    &\le  \Prob{ \One\ \middle|\ G, M}  + o(1) 
    \le \frac{1}{2} + o(1),
\end{align*}
which is a contradiction, since we have assumed that $\epsilon  >  0$ is a constant.

This completes the proof of Theoem \ref{thm:msgRandomized}.

\section{A Deterministic Algorithm} \label{sec:deter}

Our algorithm (Algorithm 2) is inspired by the solution of Afek and Gafni~\cite{Afek_1991} for the $n$-node clique. However, there are some complications that we explain below, since we cannot rely on all nodes to be connected by an edge. Note that our algorithm assumes that $n$ (or a constant factor upper bound for $\log n$) is known to all nodes.

For any node $v \in V$, we denote the degree of $v$ by $d_v$ and the ID of $v$ by $ID_v$. At any time-point in the algorithm, $L_v$ denotes the highest ID that $v$ has so far learned (among all the probe messages it has received, in the current round or in some previous round).

The algorithm proceeds as a sequence of $\Theta(\log n)$ phases. Initially every node is a ``candidate'' and is ``active''. Each node $v$ numbers its neighbors from $1$ to $d_v$, denoted by $w_{v, 1}, w_{v, 2}, \ldots, w_{v, {d_v}}$ respectively. In phase $i$, if a node $v$ is active, $v$ sends probe-messages containing its ID to its neighbors $w_{v,2^{i -1}},\dots,w_{v,k}$, where $k = \min\left\{d_v, 2^i - 1\right\}$. Each one of them replies back with the highest ID it has seen so far. If any on those ID's is higher than $ID_v$, then $v$ stops being a candidate and becomes inactive. Node $v$ also becomes inactive if it has finished sending probe-messages to all its neighbors. After finishing the $\Theta(\log n)$ phases $v$ becomes leader if it is still a candidate.

The idea behind the algorithm is to exploit the \emph{neighborhood intersection property} (cf.\ Observation \ref{main-diameter-two-property}) of diameter-$2$ networks. Since for any $u, v \in V$, there is an $x \in V$ that is connected to both $u$ and $v$ (unless $u$ and $v$ are directly connected via an edge) and acts as a ``referee'' node for candidates $u$ and $v$. This means that $x$ serves to inform $u$ and $v$ who among them is the winner, i.e., has the higher ID. Thus at the end of the algorithm, every node except the one with the highest ID should know that he is not a leader. We present the formal analysis of Theorem~\ref{thm:deterministic} in Sections ~\ref{section-deterministic-proof-of-correctness} and \ref{sec:detmsgcomplexity}.

\begin{theorem} \label{thm:deterministic}
  There exists a deterministic leader election algorithm for $n$-node networks with diameter at most $2$ that sends $O(n\log{n})$ messages and terminates in $O(\log{n})$ rounds.
\end{theorem}

In the pseudocode and the subsequent analysis we use $v$ and $ID_v$ interchangeably to denote the node $v$.

\begin{algorithm}[t]
\begin{algorithmic}[1]
	\State $v$ becomes a ``candidate'' and ``active''.
	\State $L_v \gets ID_v$.
	\State $N_v \gets ID_v$.
	\State $v$ numbers its neighbors from $1$ to $d_v$, which are called $w_{v, 1}, w_{v, 2}, \ldots, w_{v, {d_v}}$ respectively.
	\For{phase $i = 1$ to $\Theta(\log{n})$}
		\If{$v$ is active}
			\State $v$ sends a ``probe" message containing  its ID to its neighbors $w_{v,2^{i -1}},\dots,w_{v,\text{min}\left\{d_v, 2^i - 1\right\}}$.	
			\If{$d_v  \leq  2^i - 1$} \MyComment{If $v$ is finished with exploring its adjacency list, $v$ becomes inactive.}
      \label{finished-exploring-all-neighbors}
				\State $v$ becomes inactive.
			\EndIf
		\EndIf
		
		\State Let $X$ be the set (possibly empty) of neighbors of $v$ from whom $v$ receives messages in this round.
		\State Let $\mathcal{I}\mathcal{D}$ be the set of ID's sent to $v$ by the members of $X$.
		\State Let $ID_u$ be the highest ID in $\mathcal{I}\mathcal{D}$.
		
		\If{$ID_u > L_v$}\label{update-leader}
			\State $v$ sends $ID_u$ to $N_v$. \label{tell-previous-master-about-new-master}
			\State $L_v \gets ID_u$. \MyComment{$v$ stores the highest ID seen so far in $L_v$.}
			\State $N_v \gets x$. \MyComment{$v$ remembers neighbor who told $v$ about $L_v$.}
			\State $v$ becomes ``inactive'' and ``non-candidate''.
		\EndIf
		
		\State $v$ tells every member of $X$ about $L_v$, i.e., the highest ID it has seen so far. \label{tell-sender-about-current-master}
	\EndFor

	\If{$v$ is still a candidate}
		\State $v$ elects itself to be the leader.
	\EndIf
\end{algorithmic}
\caption{Deterministic Leader Election in $O(\log{n})$ rounds and with $O(n\log{n})$ messages: Code for a node $v$} \label{alg:deterministic:implicit}
\end{algorithm}

\subsection{Proof of Correctness}\label{section-deterministic-proof-of-correctness}
Define $v^{\text{max}}$ to be the node with the highest ID in $G$. 
\begin{lemma}
	$v^{\text{max}}$ becomes a leader.
\end{lemma}
\begin{proof}
	Since $v^{\text{max}}$ has the highest ID in $G$, the if-clause of Line \ref{update-leader} of Algorithm \ref{alg:deterministic:implicit} is never satisfied for $v^{\text{max}}$. Therefore $v^{\text{max}}$ never 
	becomes a non-candidate, and hence becomes a leader at the end of the algorithm.
\end{proof}

\begin{lemma}\label{deterministic-algorithm-correctness}
	No other node except $v^{\text{max}}$ becomes a leader.
\end{lemma}
\begin{proof}
	Consider any $u \in V$ such that $u \neq v^{\text{max}}$.

	\begin{itemize}
		\item \textbf{Case $1$ ($v^{\text{max}}$ and $u$ are connected via an edge):} Since $v^{\text{max}}$ has the highest ID in $G$, the if-clause of Line \ref{update-leader} of Algorithm 
		\ref{alg:deterministic:implicit} is never satisfied for $v^{\text{max}}$. Therefore $v^{\text{max}}$ becomes inactive only if it has already sent probe-messages to all its neighbors (or $v^{\text{max}}$ never becomes inactive). In particular, $u$ always receives a probe-message from $v^{\text{max}}$ containing $ID_{v^{\text{max}}}$. Since $ID_{v^{\text{max}}} > ID_u$, $u$ becomes a non-candidate at that point (if $u$ was still a candidate until that point) and therefore does not become a leader.

		\item \textbf{Case $2$ ($v^{\text{max}}$ and $u$ do not have an edge between them):} By Observation \ref{main-diameter-two-property}, there is some $x \in V$ such that both  $v^{\text{max}}$ and $u$ have edges going to $x$. And we have already established that $v^{\text{max}}$ will always send a probe-message to $x$ at some point of time or another.
		
		\begin{itemize}
            \item \textbf{Case $2$(a) ($u$ does not send a probe-message to $x$):} This implies that $u$ became inactive before it could send a probe-message to $x$. But then $u$ could have become inactive only if the if-clause of Line \ref{update-leader} of Algorithm \ref{alg:deterministic:implicit} got satisfied at some point. Then $u$ became a non-candidate too at the same time and therefore would not become a leader.
			
			\item \textbf{Case $2$(b) ($u$ sends a probe-message to $x$ before $v^{\text{max}}$ does):} Suppose $u$ sends a probe-message to $x$ at round $i$ and $v^{\text{max}}$ sends a probe-message to $x$ at round $i'$, where $1 \leq i < i' \leq \log{n}$. If $x$ had seen an ID higher than $ID_u$ up until round $i$, then $x$ immediately informs $u$ and $u$ becomes a non-candidate.
			
			So suppose not. Then, after round $i$, $x$ sets its local variables $L_x$ and $N_x$ to $ID_u$ and $u$ respectively. Let $j > i$ be the smallest integer such that $x$ receives a probe-message from a neighbor $u'$ at round $j$, where $ID_{u'} > ID_u$. Note that $v^{\text{max}}$ will always send a probe-message to $x$, therefore such a $u'$ exists. Then, after round $j$, $x$  sets its local variables $L_x$ and $N_x$ to $ID_{u'}$ and $u'$ respectively, and informs $u$ of this change. $u$ becomes a non-candidate at that point of time.
			
			\item \textbf{Case $2$(c) ($u$ and $v^{\text{max}}$ each sends a probe-message to $x$ at the same time):} Since $ID_{v^{\text{max}}}$ is the highest ID in the network, $L_x$ is assigned the value $ID_{v^{\text{max}}}$ at this point, and $x$ tells $u$ about $L_x = ID_{v^{\text{max}}} > ID_u$, causing $u$ to become a non-candidate.
			
			\item \textbf{Case $2$(d) ($u$ sends a probe-message to $x$ after $v^{\text{max}}$ does):} Suppose $v^{\text{max}}$ sends a probe-message to $x$ at round $i$ and $u$ sends a probe-message to $x$ at round $i'$, where $1 \leq i < i' \leq \log{n}$. Then $x$ sets its local variables $L_x$ and $N_x$ to $ID_{v^{\text{max}}}$ and $v^{\text{max}}$, respectively, after round $i$. So when $u$ comes probing at round $i' > i$, $x$ tells $u$ about $L_x = ID_{v^{\text{max}}} > ID_u$, causing $u$ to become a non-candidate.
		\end{itemize}
	\end{itemize}
\end{proof}

Thus Algorithm \ref{alg:deterministic:implicit} elects a unique leader (the node with the highest ID) and is therefore correct.
\subsection{Message Complexity} \label{sec:detmsgcomplexity}

\begin{lemma}\label{not-too-many-active-nodes-at-round-i}
	At the end of round $i$, there are at most $\frac{n}{2^i}$ ``active'' nodes.
\end{lemma}
\begin{proof}
	Consider a node $v$ that is active at the end of round $i$. This implies that the if-clause of Line \ref{update-leader} of Algorithm \ref{alg:deterministic:implicit} has not so far been satisfied for $v$, which in turn implies that $ID_v > ID_{w_{v, j}}$ for $1 \leq j \leq 2^i - 1$, therefore none of
	\begin{center}
		$w_{v, 1}, w_{v, 2}, \ldots, w_{v, 2^i-1}$
	\end{center}
	is active after round $i$. Thus for every active node at the end of round $i$, there are at least $2^i - 1$ inactive nodes. We call this set of inactive nodes, together with $v$ itself, the ``kingdom'' of $v$, i.e., 
	\begin{center}
		$KINGDOM(v)   \defeq   \left\{v\right\}  \cup  \left\{w_{v, 1}, w_{v, 2}, \ldots, w_{v, 2^i-1}\right\}$
	\end{center}
	\begin{center}
		and $|KINGDOM(v)| = 2^i$.
	\end{center}
	If we can show that these kingdoms are disjoint for two different active nodes, then we are done.\\
	
	\paragraph{Proof by contradiction}
	Suppose not. Suppose there are two nodes $u$ and $v$ such that
	\begin{center}
		$u \neq v$ and $KINGDOM(u) \cap KINGDOM(v)  \neq  \phi$
	\end{center}	
	(after some round $i$, $1 \leq i \leq \log{n}$). Let $x$ be such that $x  \in  KINGDOM(u) \cap KINGDOM(v)$. Since an active node obviously cannot belong to the kingdom of another active node, this $x$ equals neither $u$ nor $v$, and therefore
	\begin{center}
		$x   \in   \left\{w_{v, 1}, w_{v, 2}, \ldots, w_{v, 2^i-1}\right\}  \cap  \left\{w_{u, 1}, w_{u, 2}, \ldots, w_{u, 2^i-1}\right\}$,
	\end{center}
	that is, both $u$ and $v$ have sent their respective probe-messages to $x$. Without loss of generality, let $ID_v > ID_u$.
	
	\begin{itemize}
        \item \textbf{Case $1$ ($u$ sends a probe-message to $x$ before $v$ does):} Suppose $u$ sends a probe-message to $x$ at round $j$ and $v$ sends a probe-message to $x$ at round $j'$, where $1 \leq j < j' \leq i$. If $x$ had seen an ID higher than $ID_u$ up until round $j$, then $x$ immediately informs $u$ and $u$ becomes inactive. Contradiction.
			
		So suppose not. Then, after round $j$, $x$ sets its local variables $L_x$ and $N_x$ to $ID_u$ and $u$ respectively. Let $k > j$ be the smallest integer such that $x$ receives a probe-message from a neighbor $u'$ at round $k$, where $ID_{u'} > ID_u$. Note that $v$ sends a probe-message to $x$ at round $j'$, where $j < j' \leq i$, and $ID_v > ID_u$. Therefore such a $u'$ exists. Then, after round $k$, $x$  sets its local variables $L_x$ and $N_x$ to $ID_{u'}$ and $u'$ respectively, and informs $u$ of this change. $u$ becomes inactive at that point of time, i.e., after round $k$, where $k \leq j' \leq i$. Contradiction.
			
		\item \textbf{Case $2$ ($u$ and $v$ each sends a probe-message to $x$ at the same time):} Suppose that $u$ and $v$ each sends a probe-message to $x$ at the same round $j$, where $1 \leq j \leq i$. Since $ID_v > ID_u$, $x$ has at least one neighbor $u'$ such that $ID_{u'} > ID_u$. Therefore $x$ would not set $L_x$ to $ID_u$ (or $N_x$ to $u$), and $x$ would inform $u$ about that after round $j$, causing $u$ to then become inactive. Contradiction.
			
		\item \textbf{Case $3$ ($u$ sends a probe-message to $x$ after $v$ does):} Suppose $v$ sends a probe-message to $x$ at round $j$ and $u$ sends a probe-message to $x$ at round $j'$, where $1 \leq j < j' \leq i$. Then $x$ sets its local variables $L_x$ and $N_x$ to $ID_v$ and $v$, respectively, after round $j$. So when $u$ comes probing at round $j' > j$, $x$ tells $u$ about $L_x \geq ID_v > ID_u$, causing $u$ to become inactive. Contradiction.
	\end{itemize}
\end{proof}

\begin{lemma}
	In round $i$, at most $3n$ messages are transmitted.
\end{lemma}
\begin{proof}
	 In round $i$, each active node sends exactly $2^{i-1}$ probe messages, and each probe-message generates at most two responses (corresponding to Lines \ref{tell-previous-master-about-new-master} and 
	 \ref{tell-sender-about-current-master} of Algorithm \ref{alg:deterministic:implicit}). Thus, in round $i$, each active node contributes to, directly or indirectly, at most $3.2^{i-1}$ messages. The result 
	 immediately follows from Lemma \ref{not-too-many-active-nodes-at-round-i}.
\end{proof}

Since the algorithm runs for $\log{n}$ rounds, Theorem \ref{thm:deterministic} immediately follows.
\section{A Deterministic Lower Bound} \label{sec:detlb}

We will show a lower bound of $\Omega(n\log{n})$ message complexity by reducing the problem of ``leader election in complete graphs'' to that of ``leader election in graphs of diameter two''. This reduction itself would take two rounds and $O(n)$ messages. Then, since the former is known to have $\Omega(n\log{n})$ message complexity, the latter would have the same lower bound too (cf.\ Section \ref{sec:results}).

Suppose $\mathcal{A}$ is a leader election algorithm that works for any graph of diameter two. Let $G = (V, E)$ be our input instance for the problem of ``leader election in complete graphs'', i.e., $G$ is the complete graph on $n$ nodes, say.\\

\paragraph{The Reduction}
$G$ sparsifies itself into a diameter-two graph ($G'$, say, where $G' = (V, E')$, where $E' \subsetneq E$) on which $\mathcal{A}$ works thereafter. This sparsification takes $O(n)$ messages and a constant number of rounds (two, to be exact) and is done as follows.

\begin{itemize}
	\item \textbf{Round $1$:} Each node $v$  chooses one of its neighbours (any arbitrary one) and asks its ID. If this neighbour's ID is larger than its own ID, then v will ``drop" that edge, i..e., it won't use that 
	for communication in the subsequent simulation of $\mathcal{A}$. Otherwise $v$ will keep that edge.

	For $v \in V$, if $v$ has $\lceil\frac{n}{2}\rceil$ or more edges removed, then $v$ makes itself a ``candidate''.
	
	\item \textbf{Round $2$:} The candidates from the previous round send their ID's to all the nodes in the network using edges of $G$. By Lemma \ref{lemma-at-most-two-candidates}, there can be at most two such nodes. 
	Thus the total number of messages sent is still $O(n)$. Then each node (including the candidates themselves) receives the ID's of up to two candidates and chooses the highest of them to be the ID of the leader.
\end{itemize}

If no such node exists which has had $\lceil\frac{n}{2}\rceil$ or more edges removed, then $G'$ has diameter two (please refer to Lemma \ref{lemma-deterministic-reduction}), and we run $\mathcal{A}$ on $G'$. 
$\mathcal{A}$ returns a leader on $G'$ which makes itself the leader of $G$ too, and informs all its neighbors. This takes $O(n)$ messages.

\subsection{Proof of Correctness}

\begin{obs}\label{obs-at-most-n-1-edges-removed}
	$E$ has at most $n-1$ edges more than $E'$.
\end{obs}
\begin{proof}
	Each node except the node with the highest ID drops at most one edge. The node with the highest ID drops no edge.
\end{proof}

\begin{lemma}\label{lemma-at-most-two-candidates}
	For $n \geq 3$, there can be at most two nodes in $G'$ that has had $\lceil\frac{n}{2}\rceil$ or more edges removed.
\end{lemma}
\begin{proof}
	We consider the two cases --- when $n$ is even and when $n$ is odd --- separately in order to make the presentation simpler.\\
	
	\begin{itemize}
        \item \textbf{Case $1$: $n = 2k$ for some integer $k \geq 2$.}
	
        \paragraph{Proof by contradiction}
        Suppose that there are three or more nodes that have had $\lceil\frac{n}{2}\rceil  =  k$ or more edges removed each (either by themselves or by their neighbors). Let $u$, $v$, and $w$ be three such nodes. Since an edge is removed only if one of the incident nodes has a higher ID than the other, all of $(u, v)$, $(v, w)$, and $(w, u)$ cannot have been removed. Thus the total number of edges removed is at least $3k - 2  >  2k - 1$, which contradicts Observation \ref{obs-at-most-n-1-edges-removed}.\\

        \item \textbf{Case $2$: $n   =   2k + 1$ for some integer $k \geq 1$.}
	
        \paragraph{Proof by contradiction}
        Suppose that there are three or more nodes that have had $\lceil\frac{n}{2}\rceil   =   k + 1$ or more edges removed each (either by themselves or by their neighbors). Let $u$, $v$, and $w$ be three such nodes. Since an edge is removed only if one of the incident nodes has a higher ID than the other, all of $(u, v)$, $(v, w)$, and $(w, u)$ cannot have been removed. Thus the total number of edges removed is at least $3(k+1) - 2   >   2k$, which contradicts Observation \ref{obs-at-most-n-1-edges-removed}.
    \end{itemize}
\end{proof}

\begin{lemma}\label{lemma-deterministic-reduction}
	If no node exists in $G'$ which has had $\lceil\frac{n}{2}\rceil$ or more edges removed, then $G'$ has diameter two.
\end{lemma}
\begin{proof}
	Clearly $G'$ is not of diameter one since the node with the smallest ID in $V$ always drops at least one edge.
	
	Next we show that for any $u, v \in V$, either $u$ and $v$ are directly connected in $G'$ or $\exists w \in V$ such that $(u, w) \in E'$ and $(w, v) \in E'$.
	
	We consider the two cases --- when $n$ is even and when $n$ is odd --- separately in order to make the presentation simpler.
	
	\begin{itemize}
	
        \item \textbf{Case $1$: $n = 2k$ for some integer $k \geq 2$.}
	
        Since no node exists in $G'$ which has had $\lceil\frac{n}{2}\rceil  =  k$ or more edges removed, every node in $G'$ has degree at least $(n - 1) - (k - 1)  =  k$. Thus for any $u, v \in V$, if $(u, v) \notin E'$, then there are at least $k + k - (n - 2) = 2$ nodes in $V \setminus \left\{u, v\right\}$ that are common neighbors to both $u$ and $v$.\\

        \item \textbf{Case $2$: $n  =  2k + 1$ for some integer $k \geq 1$.}
	
        Since no node exists in $G'$ which has had $\lceil\frac{n}{2}\rceil  =  k + 1$ or more edges removed, that implies that every node in $G'$ has degree at least $(n - 1) - k  =  k$. Thus for any $u, v \in V$, if $(u, v) \notin E'$, then there is at least $k + k - (n - 2) = 1$ node in $V \setminus \left\{u, v\right\}$, which is a common neighbor to both $u$ and $v$.
    \end{itemize}
\end{proof}

\section{Conclusion} \label{sec:conc}

We settle the message complexity of leader election throughout the diameter spectrum, by presenting almost tight bounds (tight upto $\text{polylog}(n)$ factors) for diameter-two graphs which were left open by previous results \cite{Kutten_2015_TCS, Kutten_2015_JACM}. Several open problems arise from our work.

\begin{enumerate}
    \item Is it possible to show a high probability bound of $O(n)$ messages for randomized leader election that runs in $O(1)$ rounds? This will match the lower bounds, by closing the $\text{polylog}(n)$ factor. It might be possible to improve the analysis of our randomized algorithm to show $O(n\log{n})$ messages.
    
    \item Another very interesting question is whether \emph{explicit} leader election (i.e., where all nodes should also know the identity of the leader) can be performed in $\tilde{O}(n)$ messages in diameter-two graphs (this is true for complete graphs, but not for diameter three and beyond).
    
    \item The question of explicit leader election naturally begs the question whether \emph{broadcast}, another fundamental problem in distributed computing, can be solved in \emph{diameter-two} graphs with $\tilde{O}(n)$ messages and $O(\text{polylog}(n))$ rounds if $n$ is known.
    \footnote{In contrast, we note that $\Omega(m)$ is a lower bound for broadcast on graphs of \emph{diameter at least three}, even if $n$ is known and even for randomized algorithms \cite{Kutten_2015_JACM}.}
    
    \item Removing the assumption of the knowledge of $n$ (or showing that it is not possible) for deterministic algorithms with $\tilde O(n)$ message complexity and running in $\tilde{O}(1)$ rounds is open as well.
\end{enumerate}


\newpage

\bibliography{diameter_two_leader_election_references}
\bibliographystyle{plain}

\newpage
\appendix
\section*{Appendix}

\begin{proof}[Proof of Lemma \ref{lemma-Lagrangian-Hessian}]
	We define the Lagrangian as
	\begin{equation}\label{eqn-Lagrangian-definition}
		\mathcal{L}  \defeq  f(x_1, x_2, \ldots, x_n) - \lambda(\sum_{i=1}^n x_i - C)
	\end{equation}
	where $\lambda$ is the Lagrange multiplier. We can find the \emph{critical points} of the Lagrangian by solving the set of equations
	\begin{equation}\label{eqn-Lagrangian-1}
		\frac{\partial f}{\partial x_i}  =  \lambda \frac{\partial \sum_{i = 1}^n x_i}{\partial x_i}\text{ for }i = 1, 2, \ldots, n
	\end{equation}
	and
	\begin{equation}\label{eqn-constraint}
		\sum_{i=1}^n x_i  =  C
	\end{equation}
	Simplifying Equation \ref{eqn-Lagrangian-1}, we get
	\begin{equation}\label{eqn-Lagrangian-2}
		-\frac{\log{x_i}}{x_i^2}  =  \lambda\text{ for }i = 1, 2, \ldots, n
	\end{equation}

	One possible (feasible) solution of Equations \ref{eqn-Lagrangian-2} and \ref{eqn-constraint} is,
	\begin{equation}\label{eqn-solution}
		x_i  =  \frac{C}{n}\text{ for all }1 \leq i \leq n
	\end{equation}
	and
	\begin{equation}\label{eqn-solution-lambda}
		\lambda^*  =  -\frac{\log{(\frac{C}{n})}}{(\frac{C}{n})^2}
	\end{equation}
	Let $X^*$ be a vector of dimension $n$ defined by $X^*  \defeq  (\frac{C}{n}, \frac{C}{n}, \ldots, \frac{C}{n})$. Then we have already shown that $X^*$ and $\lambda^*$ are a \emph{critical point} for the Lagrange function $\mathcal{L}$. We claim that $X^*$ is also a local minima for $f(x)$ under the constraint of Equation \ref{eqn-constraint}.\\
	
	We show that by constructing the Bordered Hessian matrix $H^B$ of the Lagrange function. Let $L^*_{ij}   \defeq   \left.\frac{\partial}{\partial x_j} (\frac{\partial \mathcal{L}}{\partial x_i})\right\rvert_{X^*}$, where $\mathcal{L}$ is the Lagrange function as defined in Equation \ref{eqn-Lagrangian-definition}. Then
	\[
	H^B  =
	\begin{bmatrix}
		0 & 1 & 1 & \cdots & 1 \\
		1 & L^*_{11} & L^*_{12} & \cdots & L^*_{1n}\\
		1 & L^*_{21} & L^*_{22} & \cdots & L^*_{2n}\\
		\vdots & \vdots & \vdots & \ddots & \vdots\\
		1 & L^*_{n1} & L^*_{n2} & \cdots &L^*_{nn}
	\end{bmatrix}
	\]
	We note that $L^*_{ii} = \frac{2\log{(\frac{C}{n})} - 1}{(\frac{C}{n})^3}  -  \lambda^*$ for all $1 \leq i \leq n$, and $L^*_{ij} = 0$ for all $(i, j)$ such that $i \neq j$. Hence
	
	\[
	H^B  =
	\begin{bmatrix}
		0 & 1 & 1 & \cdots & 1 \\
		1 & \frac{2\log{(\frac{C}{n})} - 1}{(\frac{C}{n})^3}  -  \lambda^* & 0 & \cdots & 0\\
		1 & 0 & \frac{2\log{(\frac{C}{n})} - 1}{(\frac{C}{n})^3}  -  \lambda^* & \cdots & 0\\
		\vdots & \vdots & \vdots & \ddots & \vdots\\
		1 & 0 & 0 & \cdots & \frac{2\log{(\frac{C}{n})} - 1}{(\frac{C}{n})^3}  -  \lambda^*
	\end{bmatrix}
	\]
	We show that $H^B$ is \emph{positive definite} (which is a sufficient condition for $X^*$ to be a local minima) by checking the signs of the leading principal minors. For any $1 \leq i \leq n$, $|H^B_i|$ is the determinant of a square matrix of dimension $i+1$, and is given by
	\[
	|H^B_i|  =
	\begin{vmatrix}
		0 & 1 & 1 & \cdots & 1 \\
		1 & \frac{2\log{(\frac{C}{n})} - 1}{(\frac{C}{n})^3}  -  \lambda^* & 0 & \cdots & 0\\
		1 & 0 & \frac{2\log{(\frac{C}{n})} - 1}{(\frac{C}{n})^3}  -  \lambda^* & \cdots & 0\\
		\vdots & \vdots & \vdots & \ddots & \vdots\\
		1 & 0 & 0 & \cdots & \frac{2\log{(\frac{C}{n})} - 1}{(\frac{C}{n})^3}  -  \lambda^*
	\end{vmatrix}
	\] \\
	$=  -i(\frac{2\log{(\frac{C}{n})} - 1}{(\frac{C}{n})^3}  -  \lambda^*)^{i-1}$.
	
	\begin{align*}
		&\text{But }\frac{2\log{(\frac{C}{n})} - 1}{(\frac{C}{n})^3}   >   0 \tag{since $C \geq n\sqrt{2}$, $2\log{(\frac{C}{n})} - 1   >   0$}\\
		&\text{and }\lambda^*  =  -\frac{\log{(\frac{C}{n})}}{(\frac{C}{n})^2}  <  0\text{.}
    \end{align*}
    
    Hence
    \begin{align*}
		&\frac{2\log{(\frac{C}{n})} - 1}{(\frac{C}{n})^3}  -  \lambda^*  >  0\\
		&\implies     -i(\frac{2\log{(\frac{C}{n})} - 1}{(\frac{C}{n})^3}  -  \lambda^*)^{i-1}   <   0\\
		&\text{i.e., }|H^B_i|  <  0\text{ for all }1 \leq i \leq n\\
		&\implies  H^B\text{ is positive definite.}
	\end{align*}
\end{proof}

\end{document}